\documentclass[12pt,reqno]{amsart}
\usepackage{amsmath,amssymb,mathrsfs,amsopn,amsthm,xcolor,bm,times}
\usepackage[colorlinks=true,citecolor=blue,linkcolor=purple
]{hyperref}
\usepackage[final]{showlabels}

\usepackage[mathscr]{euscript}
\usepackage{geometry}
\numberwithin{equation}{section}
\newtheorem{thm}{Theorem}[section]
\newtheorem{prop}[thm]{Proposition}
\newtheorem{lem}[thm]{Lemma}
\newtheorem{cor}[thm]{Corollary}
\newtheorem{conj}[thm]{Conjecture}
\theoremstyle{definition} 
\newtheorem{eg}[thm]{Example}

\theoremstyle{remark}
\newtheorem{rem}[thm]{Remark}

\newcommand{\beq}{\begin{equation}}
\newcommand{\eeq}{\end{equation}}
\newcommand{\be}{\begin{equation*}}
\newcommand{\ee}{\end{equation*}}
\newcommand{\bs}{\boldsymbol}

\newcommand{\C}{\mathbb{C}}

\newcommand{\Z}{\mathbb{Z}}

\newcommand{\bB}{\mathbb{B}}

\newcommand{\mc}{\mathcal}

\newcommand{\cT}{\mathcal{T}}

\newcommand{\g}{\mathfrak{gl}(1|1)}
\newcommand{\gl}{\mathfrak{gl}}
\newcommand{\h}{\mathfrak{h}}

\newcommand{\fkS}{\mathfrak{S}}

\newcommand{\ugt}{\mathrm{U}(\mathfrak{gl}(1|1)[t])}

\newcommand{\End}{\mathrm{End}}

\newcommand{\sing}{{\mathrm{sing}}}   
\newcommand{\ch}{{\mathrm{ch}}}

\newcommand{\pa}{\partial}

\newcommand{\gge}{\geqslant}
\newcommand{\lle}{\leqslant}
\newcommand{\la}{\lambda}

\newcommand{\bla}{{\bm\lambda}}
\newcommand{\bLa}{{\bm\Lambda}}
\newcommand{\glMN}{\mathfrak{gl}(m|n)}

\newcommand{\Uone}{\mathrm{U}(\mathfrak{gl}_{1|1})}

\newcommand{\bmx}{\begin{pmatrix}}    
\newcommand{\emx}{\end{pmatrix}}

\newcommand{\qedd}{\tag*{$\square$}}
\newcommand{\vSi}{\varSigma}

\begin{document}
\pagestyle{myheadings}
\setcounter{page}{1}

\title[Completeness of Bethe ansatz for Gaudin models]{Completeness of Bethe ansatz for Gaudin models\\ associated with $\gl(1|1)$}

\author{Kang Lu}
\address{K.L.: Department of Mathematics, University of Denver, 
\newline
\strut\kern\parindent 2390 S. York St., Denver, CO 80208, USA}\email{kang.lu@du.edu}

\maketitle
\vspace{-0.8cm}
\begin{abstract} We study the Gaudin models associated with $\gl(1|1)$. We give an explicit description of the algebra of Hamiltonians (Gaudin Hamiltonians) acting on tensor products of polynomial evaluation $\gl(1|1)[t]$-modules. It follows that there exists a bijection between common eigenvectors (up to proportionality) of the algebra of Hamiltonians and monic divisors of an explicit polynomial written in terms of the highest weights and evaluation parameters. In particular, our result implies that each common eigenspace of the algebra of Hamiltonians has dimension one. Therefore, we confirm \cite[Conjecture 8.3]{HMVY}. We also give dimensions of the generalized eigenspaces. Moreover, we express the generating pseudo-differential operator of Gaudin transfer matrices associated to antisymmetrizers in terms of the quadratic Gaudin transfer matrix and the center of $\ugt$.
\medskip

\end{abstract}

\thispagestyle{empty}

\section{Introduction}
Supersymmetric integrable models have been studied extensively for the last 4 decades since their introduction back to 1980s \cite{KS,K}. However, the vast majority of work are done for spin chains while there is far less information about Gaudin models associated to Lie superalgebras.

This paper is devoted to the study of the Gaudin models associated to the simplest Lie superalgebra $\g$. Our main motivation is to prove the Bethe ansatz conjecture for the $\g$ Gaudin models defined on tensor products of polynomial modules, which is formulated in the form of \cite[Conjecture 8.3]{HMVY}. We expect the results of this paper is important to understand the completeness of Bethe ansatz for $\glMN$ Gaudin models, see \cite[Conjecture 8.2]{HMVY} and cf. \cite{CLV20,MNV}.

Our main tool is the statements of Bethe ansatz for $\g$ Gaudin models that are established in the generic situation, see \cite{MVY15}. For generic parameters or more specially in the case when the Bethe ansatz equation has no multiple roots, the Bethe ansatz describes completely the spectrum of the Gaudin Hamiltonians, and the Gaudin Hamiltonians are simultaneously diagonalizable with simple spectrum. The more subtle case is when the Bethe ansatz equation has roots of non-trivial multiplicity, then the Gaudin Hamiltonians develop Jordan blocks. To attack this, we describe the image of algebra of Hamiltonians (Bethe algebra) explicitly and show that the Gaudin model for $\g$ is perfectly integrable, see \cite{Lu20}. The perfect integrability introduced in \cite{Lu20} is motivated by \cite{MTV09} to study the completeness of Bethe ansatz. Here by perfect integrability, we mean the algebra of Hamiltonians acts on the Hilbert space $\mathbb V$ cyclically and the image of the algebra of Hamiltonians in $\End(\mathbb V)$ is a Frobenius algebra.

\medskip 

Let us discuss our findings in more detail. 

We consider tensor products $L_{\bLa}(\bm b)=\bigotimes_{s=1}^k L_{\bla^{(s)}}(b_s)$ of polynomial evaluation modules of $\ugt$, where $\bLa=(\bla^{(1)},\dots,\bla^{(k)})$ is a sequence of polynomial $\g$-weights and $\bm b=(b_1,\dots,b_k)$ is a sequence of distinct complex numbers. Set $\bla^{(s)}=(\alpha_s,\beta_s)$.

Define the Bethe algebra to be the unital subalgebra of $\ugt$ generated by the coefficients of the (quadratic) Gaudin transfer matrix and the center of $\ugt$. Here the Gaudin transfer matrix, see \eqref{eq:H-transfer}, can be thought as the generating function of Gaudin Hamiltonians, see \eqref{eq:H-gen}. We need to find the spectrum of the Gaudin transfer matrix acting on the subspace $(L_{\bLa}(\bm b))_{(n-l,l)}^\sing$ of singular vectors in $L_{\bLa}(\bm b)$ of weight $(n-l,l)$ where $n=\sum_{s=1}^k(\alpha_s+\beta_s)$ and $0\lle l<k$. Indeed, we obtain more than that. Let us list our main results about completeness of Bethe ansatz, see Theorem \ref{thm tensor irr}.
\begin{itemize}
\item We prove that the Bethe algebra acts on $(L_{\bLa}(\bm b))_{(n-l,l)}^\sing$ cyclically and hence its image in $\End((L_{\bLa}(\bm b))_{(n-l,l)}^\sing)$ has dimension ${k-1}\choose{l}$.
    \item We further show that the image of the Bethe algebra in $\End((L_{\bLa}(\bm b))_{(n-l,l)}^\sing)$ is isomorphic to
\[
\C[w_1,\dots,w_{k-1}]^{\fkS_{l}\times \fkS_{k-l-1}}/\langle n\prod_{i=1}^k(x-w_i)-\varphi_{\bLa,\bm b}(x)\rangle,\]
where
$$
\varphi_{\bLa,\bm b}(x):=\prod_{s=1}^k(x-b_s)\sum_{s=1}^k\frac{\alpha_s+\beta_s}{x-b_s}.
$$
In particular, the image of the Bethe algebra is a Frobenius algebra. Thus, we establish the perfect integrability of the Gaudin models for $\g$ in the sense of \cite{Lu20}.
\item Consequently, we obtain that the eigenvectors (up to proportionality) of the Bethe algebra in $(L_{\bLa}(\bm b))^\sing_{(n-l,l)}$ are in a bijective correspondence with the monic 
polynomials of degree $l$ which divide $\varphi_{\bLa,\bm b}(x)$. And to each monic divisor $y$ we have exactly one eigenvector of the Bethe algebra and a generalized eigenspace of dimension given by the product of binomial coefficients
$$\prod_{a\in\C} {{\textrm{Mult}_a(\varphi_{\bLa,\bm b}(x))}\choose{\textrm{Mult}_a(y(x))}},$$ 
where $\textrm{Mult}_a(f)$ denotes the order of zero of $f(x)$ at $x=a$.
\end{itemize}

The proof of aforementioned statements are based on the idea of \cite{MTV09}, cf. \cite{LM21}. Note that the perfect integrability for the inhomogeneous $\glMN$ Gaudin models with diagonal twists defined on tensor products of symmetric powers of the vector representations was established in \cite[Corollary 5.3]{HM20} by studying duality between $\glMN$ and $\gl(k)$ Gaudin models and using the known results from \cite{MTV08}.

\medskip

Finally, we give a description of the eigenvalues/eigenvectors of the Bethe algebra in terms of "opers" under the philosophy of the standard geometric Langlands. Given a monic divisor $y$ of $\varphi_{\bLa,\bm b}(x)$, or, in other words, an eigenvector $v_y$ of the Bethe algebra, following \cite{HMVY}, we have a pseudo-differential operator $\mc {D}_y$, see the right hand side of \eqref{eq:eigen-D}. From the explicit formula for the eigenvalue, one sees that the coefficients of the pseudo-differential operator in this case are essentially eigenvalues of the first two Gaudin transfer matrices acting on $v_y$. Again, we improve on that as follows. Let $L(x)=((-1)^{|i|}e_{ij}(x))_{i,j=1,2}$ be the generating matrix of the algebra $\ugt$. By \cite{MR14}, the Berezinian  $\textrm{Ber}(\pa_x-L(x))$ is a generating function for the Gaudin transfer matrices, see \eqref{eq diff trans} and \eqref{eq:Ber-Gaudin}. By $\textrm{Ber}(\pa_x-L(x))v_y=\mathcal{D}_y v_y$, see Theorem \ref{thm:D-eigenvalue-b}, it follows that there is a universal formula for the pseudo-differential operator in terms of the first two Gaudin transfer matrices, which produces $\mc D_y$ when applied to the vector $v_y$ for all $y$, see Corollary \ref{universal oper}.

\medskip 

The theory of Bethe ansatz has been developed for Gaudin models associated with $\mathfrak{osp}(1|2)$ in \cite{KM01}. We plan to further study the completeness of Bethe ansatz for $\mathfrak{osp}(1|2)$ Gaudin models, cf. \cite{LM:2021}.

\medskip 

The paper is organized as follows. In Section \ref{sec notation}, we fix notations and discuss basic facts of the algebraic Bethe ansatz for $\g$ Gaudin models. Then we study the space $\mc V^{\fkS}$ and Weyl modules, and their properties in Section \ref{sec space VS}. Section \ref{sec main thms} contains the main theorems where we also discuss the  higher Gaudin transfer matrices and the relations between higher Gaudin transfer matrices and the first two Gaudin transfer matrices. Section \ref{sec proof} is dedicated to the proofs of main theorems. 

\medskip

{\bf Acknowledgments.} The author thanks C.-L. Huang and E. Mukhin for interesting discussions. 

\section{Preliminaries}\label{sec notation}
\subsection{Lie superalgebra $\g$ and its representations}
A \emph{vector superspace} $V = V_{\bar 0}\oplus V_{\bar 1}$ is a $\Z_2$-graded vector space. Elements of $V_{\bar 0}$ are called \emph{even}; elements of
$V_{\bar 1}$ are called \emph{odd}. We write $|v|\in\{\bar 0,\bar 1\}$ for the parity of a homogeneous element $v\in V$. Set $(-1)^{\bar 0}=1$ and $(-1)^{\bar 1}=-1$.

Consider the vector superspace $\C^{1|1}$, where $\dim(\C^{1|1}_{\bar 0})=1$ and  $\dim(\C^{1|1}_{\bar 1})=1$. We choose a homogeneous basis $v_1,v_2$ of $\C^{1|1}$ such that $|v_1|=\bar 0$ and $|v_2|=\bar 1$. For brevity we shall write their
parities as $|v_i|=|i|$. Denote by $E_{ij}\in\End(\C^{1|1})$ the linear operator of parity $|i|+|j|$ such that $E_{ij}v_r=\delta_{jr}v_i$ for $i,j,r=1,2$. 

The Lie superalgebra $\g$ is spanned by elements $e_{ij}$, $i,j=1, 2$, with parities $e_{ij}=|i|+|j|$ and the supercommutator relations are given by
\[
[e_{ij},e_{rs}]=\delta_{jr}e_{is}-(-1)^{(|i|+|j|)(|r|+|s|)}\delta_{is}e_{rj}.
\]Let $\h$ be the commutative Lie subalgebra of $\g$ spanned by $e_{11},e_{22}$. Denote the universal enveloping algebras of $\gl_{1|1}$ and $\h$ by $\Uone$ and $\mathrm{U}(\h)$, respectively.

We call a pair $\bla=(\la_1,\la_2)$ of complex numbers a $\g$-\emph{weight}. Set $|\bla|=\la_1+\la_2$. A $\g$-weight $\bla$ is \textit{non-degenerate} if $\la_1+\la_2\ne 0$.

Let $M$ be a $\g$-module. A non-zero vector $v\in M$ is called \emph{singular} if $e_{12}v=0$. Denote the subspace of all singular vectors of $M$ by $(M)^{\rm sing}$.  A non-zero vector $v\in M$ is called \emph{of weight} $\bla=(\la_1,\la_2)$ if $e_{11}v=\la_1 v$ and $e_{22}v=\la_2 v$. Denote by $(M)_\bla$ the subspace of $M$ spanned by vectors of weight $\bla$. Set $(M)^{\sing}_\bla=(M)^{\sing}\cap (M)_\bla$.

Let $\bLa=(\bla^{(1)},\dots,\bla^{(k)})$ be a sequence of $\g$-weights. Set $|\bLa|=\sum_{s=1}^k|\bla^{(s)}|$.

Denote by $L_\bla$ the irreducible $\g$-module generated by an even singular vector $v_\bla$ of weight $\bla$. Then $L_{\bla}$ is two-dimensional if $\bla$ is non-degenerate and one-dimensional otherwise. Clearly, $\C^{1|1}\cong L_{\omega_1}$, where $\omega_1=(1,0)$. 

A $\g$-module $M$ is called a \emph{polynomial module} if $M$ is a submodule of $(\C^{1|1})^{\otimes n}$ for some $n\in \Z_{\gge 0}$. We say that $\bla$ is a \emph{polynomial weight} if $L_\bla$ is a polynomial module. Weight $\bla=(\la_1,\la_2)$ is a polynomial weight if and only if $\la_1,\la_2\in \Z_{\gge 0}$ and either $\la_1>0$ or $\la_1=\la_2=0$. We also write $L_{(\la_1,\la_2)}$ for $L_{\bla}$.

For non-degenerate polynomial weights $\bla=(\la_1,\la_2)$ and $\bm\mu=(\mu_1,\mu_2)$, we have
\[
L_{(\la_1,\la_2)}\otimes L_{(\mu_1,\mu_2)}=L_{(\la_1+\mu_1,\la_2+\mu_2)}\oplus L_{(\la_1+\mu_1-1,\la_2+\mu_2+1)}.
\]

\subsection{Current superalgebra $\g[t]$}
Denote by $\g[t]$ the Lie superalgebra $\g\otimes\C[t]$ of $\gl_{1|1}$-valued polynomials with the point-wise supercommutator. Call $\g[t]$ the \emph{current superalgebra} of $\g$. We identify $\g$ with the subalgebra $\g\otimes 1$ of constant polynomials in $\g[t]$.

We write $e_{ij}[r]$ for $e_{ij}\otimes t^r$, $r\in \Z_{\gge 0}$. A basis of $\g[t]$ is given by $e_{ij}[r]$, $i,j=1,2$ and $r\in \Z_{\gge 0}$. They satisfy the supercommutator relations
\[
[e_{ij}[r],e_{kl}[s]]=\delta_{jk}e_{il}[r+s]-(-1)^{(|i|+|j|)(|k|+|l|)}\delta_{il}e_{kj}[r+s].
\]In particular, one has 
\beq\label{eq:anti}
(e_{12}[r])^2=(e_{21}[r])^2=0,\quad e_{21}[r]e_{21}[s]=-e_{21}[s]e_{21}[r]
\eeq
in the universal enveloping superalgebra $\mathrm{U}(\g[t])$. The universal enveloping superalgebra $\mathrm{U}(\g[t])$ is a Hopf superalgebra with the coproduct given by
\[
\Delta(X)=X\otimes 1+1\otimes X,\quad \text{for }\ X\in \g[t].
\]

Let $e_{ij}(x)=\sum_{r=0}^\infty e_{ij}[r]x^{-r-1}$, where $x$ is a formal variable. For each $a\in \C$, there exists an automorphism of $\ugt$, $\rho_a:e_{ij}(x)\to e_{ij}(x-a)$. Given a $\g[t]$-module $M$, denote by $M(a)$ the pull-back of $M$ through the automorphism $\rho_a$.

For each $a\in \C$, we have the evaluation map $\mathsf{ev}_a: e_{ij}(x)\mapsto e_{ij}/(x-a)$. For a $\g$-module $L$, denote by $L(a)$ the $\g[t]$-module obtained by pulling back $L$ through the evaluation map $\mathsf{ev}_a$. We call $L(a)$ an {\it evaluation module} at $a$. 

Given any series $\zeta(x)\in x^{-1}\C[x^{-1}]$, we have the one-dimensional $\g[t]$-module generated by an even vector $v$ satisfying $e_{ij}(x)v=\delta_{ij}(-1)^{|j|}\zeta(x)v$. We denote this module by $\C_{\zeta}$.

If $b_1,\dots,b_n$ are pairwise distinct complex numbers and $L_1,\dots,L_n$ are finite-dimensional irreducible $\g$-modules, then the $\g[t]$-module $\bigotimes_{s=1}^n L_s(b_s)$ is irreducible.

There is a natural $\Z_{\gge 0}$-gradation on $\mathrm{U}(\g[t])$ such that $\deg(e_{ij}[r])=r$. Let $M$ be a $\Z_{\gge 0}$-graded space with finite-dimensional homogeneous components. Let $M_j\subset M$ be the homogeneous component of degree $j$. We call the formal power series in variable $q$,
\beq\label{eq grade ch}
\ch(M)=\sum_{j=0}^\infty \dim(M_j)\,q^j
\eeq
the \emph{graded character} of $M$.

\subsection{Gaudin Hamiltonians and Bethe ansatz}\label{sec BA}
In this section, we recall the Gaudin Hamiltonians and the corresponding Bethe ansatz from \cite{MVY15}.

Let $\bLa=(\bla^{(1)},\dots,\bla^{(k)})$ be a sequence of polynomial $\g$-weights and $\bm b=(b_1,\dots,b_k)$ a sequence of distinct complex numbers, where $\bla^{(s)}=(\alpha_s,\beta_s)$. Set $n=|\bLa|=\sum_{s=1}^k(\alpha_s+\beta_s)$ and $L_{\bLa}=\bigotimes_{s=1}^k L_{\bla^{(s)}}$. The {\it quadratic Gaudin Hamiltonians} are the linear maps $\mc H_r \in \End(L_{\bLa})$ given by
\beq\label{eq:Gaudin-H}
\mc H_r:=\sum_{s=1}^k \frac{e_{11}^{(r)}e_{11}^{(s)}-e_{12}^{(r)}e_{21}^{(s)}+e_{21}^{(r)}e_{12}^{(s)}-e_{22}^{(r)}e_{22}^{(s)}}{b_r-b_s},\quad 1\lle r\lle k.
\eeq
where $e_{ab}^{(r)}=1^{\otimes(r-1)}\otimes e_{ab}\otimes 1^{\otimes(k-r)}$.
\begin{lem}[{\cite[Proposition 3.1]{MVY15}}]
The Gaudin Hamiltonians $\mc H_r$
\begin{enumerate}
    \item are mutually commuting: $[\mc H_r,\mc H_s]=0$ for all $r,s$;
    \item commute with the action of $\g$: $[\mc H_r,X]=0$ for all $r$ and $X\in \g$.
\end{enumerate}
\end{lem}

Instead of working on Gaudin Hamiltonians $\mc H_s$, we work on the generating function of Gaudin Hamiltonians,
\beq\label{eq:H-transfer}
\mathscr H(x):=\sum_{r=2}^\infty \mathscr H_rx^{-r}=\frac{1}{2}\sum_{a,b=1}^2 e_{ab}(x)e_{ba}(x)(-1)^{|b|}.
\eeq
The operator $\mathscr H(x)$ acts on the tensor product of the evaluation $\g[t]$-modules 
$$
L_{\bLa}(\bm b):=\bigotimes_{s=1}^k L_{\bla^{(s)}}(b_s).
$$ Note that $L_{\bLa}(\bm b)$ and $L_{\bLa}$ are isomorphic as $\g$-modules via the identity map, then we have
\beq\label{eq:H-gen}
\mathscr H(x)=\frac{1}{2}\sum_{s=1}^k\frac{\alpha_s(\alpha_s-1)-\beta_s(\beta_s+1)}{(x-b_s)^2}\mathrm{Id}+\sum_{s=1}^k\frac{1}{x-b_s}\mc H_s,
\eeq
as operators in $\End(L_{\bLa})=\End(L_{\bLa}(\bm b))$. We call $\mc H(x)$ the {\it Gaudin transfer matrix}.

We are interested in finding the eigenvalues and eigenvectors of the Gaudin transfer matrix in $L_{\bLa}(\bm b)$. To be more precise, we call \beq\label{eq series}
\xi(x)= \sum_{r=2}^\infty \xi_rx^{-r},\qquad \xi_r\in \C,
\eeq 
an \emph{eigenvalue} of $\mathscr H(x)$ if there exists a non-zero vector $v\in L_{\bLa}(\bm b)$ such that $\mathscr H_rv=\xi_r v$ for all $r\in\Z_{>1}$. If $\xi(x)$ is a rational function, we consider it as a power series in $x^{-1}$ as \eqref{eq series}. The vector $v$ is called an \textit{eigenvector} of $\mathscr H(x)$ corresponding to eigenvalue $\xi(x)$. We also define the \emph{eigenspace of $\mathscr H(x)$ in $L_{\bLa}(\bm b)$ corresponding to eigenvalue $\xi(x)$} as $\bigcap_{r=2}^\infty \ker(\mathscr H_r|_{L_{\bLa}(\bm b)}-\xi_r)$.

\medskip 

It is sufficient to consider $L_\bLa$ with $\beta_s=0$ for all $s$. Indeed, if $L_{\bLa}(\bm b)$ is an arbitrary tensor product and 
\[
\xi(x)=\sum_{s=1}^k\frac{\beta_s}{x-b_s},
\]
then 
$$
L_{\bLa}(\bm b)\otimes \C_{\xi}\cong L_{\widetilde\bLa}(\bm b), \quad \tilde{\bla}^{(s)}=(\alpha_s+\beta_s,0).
$$

Identify $L_{\bLa}(\bm b)\otimes \C_{\xi}$ with $L_{\bLa}(\bm b)$ as vector spaces. 
Then $\mathscr H(x)$ acting on $L_{\bLa}(\bm b)\otimes \C_{\xi}$ coincides with $\mathscr H(x)+2\zeta(x)(e_{11}(x)+e_{22}(x))$ acting on  $L_{\bLa}(\bm b)$. Note that the coefficients of $e_{11}(x)+e_{22}(x)$ are central in $\ugt$ and hence acts on $L_{\bLa}(\bm b)$ by scalars,  therefore the problem of diagonalization of the Gaudin transfer matrix in $L_{\bLa}(\bm b)$ is reduced to diagonalization of the Gaudin transfer matrix in $L_{\widetilde\bLa}(\bm b)$.

Since $L_{\bla}$ is one-dimensional if $\bla$ is degenerate, similarly, it suffices to consider the case that all participant $\g$-weights are non-degenerate. Hence, we shall always assume throughout the paper that $\bla^{(s)}$ are non-degenerate for all $1\lle s\lle k$.

\medskip

The main method to find eigenvalues and eigenvectors of the Gaudin transfer matrix in $L_{\bLa}$ is the algebraic Bethe ansatz which we recall from \cite{MVY15}. 

Fix a non-negative integer $l$. Let $\bm t=(t_1,\dots,t_l)$ be a sequence of complex numbers. Define the polynomial $y_{\bm t}=\prod_{i=1}^l(x-t_i)$. We say that polynomial $y_{\bm t}$ \emph{represents} $\bm t$.

Set 
\beq\label{eq:zeta}
\zeta_{\bLa,\bm b}(x):=\sum_{s=1}^k\frac{\alpha_s+\beta_s}{x-b_s}.
\eeq

A sequence of complex numbers $\bm t$ is called a \emph{solution to the Bethe ansatz equation associated to} $\bLa$, $\bm b$, $l$ if
\beq\label{eq gl11 BAE}
y_{\bm t}(x) \quad\text{ divides the polynomial }\quad \varphi_{\bLa,\bm b}(x):=\zeta_{\bLa,\bm b}(x)\prod_{s=1}^k(x-b_s).
\eeq
We do not distinguish solutions which differ by a permutation of coordinates (that is represented by the same polynomial). 

Let $v_s$ be the highest weight vector of $L_{\bla^{(s)}}$, and set $|0\rangle =v_1\otimes\cdots\otimes v_k$. We call $|0\rangle$ the {\it vacuum vector}.

Define the \emph{off-shell Bethe vector} $\bB_l(\bm t)\in (L_\bLa)_{(n-l,l)}$ by
\beq\label{eq bv gl11}
\bB_l(\bm t)=e_{12}(t_1)\cdots e_{12}(t_l)\,|0\rangle.
\eeq
Since $e_{12}(x)e_{12}(u)=-e_{12}(u)e_{12}(x)$, the order of $t_i$ is not important. Moreover, the off-shell Bethe vector is zero if $t_i=t_j$ for some $1\lle i\ne j\lle l$. 

If $\bm t$ is a solution of the Bethe ansatz equation \eqref{eq gl11 BAE}, we call $\bB_l(\bm t)$ an \emph{on-shell Bethe vector}.

Let $\bm t$ be a solution of the Bethe ansatz equation associated to $\bLa$, $\bm b$, $l$.  The following statements are known, see \cite[Section VI]{MVY15}.
\begin{thm}[\cite{MVY15}]\label{thm gl11 eigenvalue in y}
If the on-shell Bethe vector ${\bB}_l(\bm t)$ is non-zero, then ${\bB}_l(\bm t)$ is an eigenvector of the Gaudin transfer matrix $\mathscr H(x)$ with the corresponding eigenvalue 
\beq\label{eq gl11 eigenvalue in y}
\mc E_{y_{\bm t},\bLa,\bm b}(x)=\frac{1}{2}\zeta_{\bLa,\bm b}'(x)-\zeta_{\bLa,\bm b}(x)\frac{y_{\bm t}'(x)}{y_{\bm t}(x)}+\sum_{r,s=1}^k\frac{\alpha_r\alpha_s-\beta_r\beta_s}{2(x-b_r)(x-b_s)}.
\eeq
where $\zeta_{\bLa,\bm b}(x)$ is given by \eqref{eq:zeta}.\qed
\end{thm}

Consider another Gaudin transfer matrix, see Example \ref{eg:transfer},
\beq\label{eq:T-transfer}
\mc T(x)=\frac{1}{2}\big(\dot{e}_{11}(x)+\dot{e}_{22}(x)\big)+\frac{1}{2}\big(e_{11}(x)+e_{22}(x)\big)^2-\mathscr H(x),
\eeq
where $\dot{e}_{ii}(x)=\pa_x(e_{ii}(x))$, $i=1,2$. Then the eigenvalue of $\mc T(x)$ acting on the on-shell Bethe vector $\mathbb B_l(\bm t)$ is 
\beq\label{eq:T-eigenvalue}
\begin{split}
\mathscr E_{y_{\bm t},\bLa,\bm b}(x)=&\ \zeta_{\bLa,\bm b}(x)\frac{y_{\bm t}'(x)}{y_{\bm t}(x)}+\sum_{r,s=1}^k\frac{\alpha_r\beta_s+\alpha_s\beta_r+2\beta_r\beta_s}{2(x-b_r)(x-b_s)}\\ =&\ \zeta_{\bLa,\bm b}(x)\Big(\frac{y_{\bm t}'(x)}{y_{\bm t}(x)}+\sum_{s=1}^k\frac{\beta_s}{x-b_s}\Big).    
\end{split}
\eeq

\begin{prop}[\cite{MVY15}]\label{prop bv sing}
The on-shell Bethe vector $ {\bB}_l(\bm t)$ is singular.\qed
\end{prop}

It is important to know if the on-shell Bethe vectors are non-zero. 

\begin{prop}[\cite{MVY15}]\label{prop bv nonzero}
Suppose the polynomial $\varphi_{\bLa,\bm b}(x)$ only has simple roots, then the on-shell Bethe vector $\mathbb B_l(\bm t)$ is nonzero.\qed
\end{prop}

The conjecture of completeness of Bethe ansatz for Gaudin models associated with $\g$ was reformulated in \cite[Conjecture 8.2]{HMVY} as follows.
\begin{conj}
Suppose all weights $\bla^{(s)}$, $1\lle s\lle k$, are polynomial $\g$-weights. Then the Gaudin transfer matrix $\mathscr H(x)$ has a simple spectrum in $(L_{\bLa}(\bm b))^\sing$. There exists a bijectively correspondence between the monic divisors $y$ of the polynomial $\varphi_{\bLa,\bm b}$ and the eigenvectors $v$ of the
Gaudin transfer matrices (up to multiplication by a non-zero
constant). Moreover, this bijection is such that $\mathscr H(x)v = \mc E_{y,\bLa,\bm b}(x)v$, where $\mc E_{y,\bLa,\bm b}(x)$
is given by \eqref{eq gl11 eigenvalue in y}.\qed
\end{conj}
By simple spectrum we mean that if $v_1$, $v_2$ are eigenvectors of $\mathscr H(x)$ and
$v_1\ne cv_2$, $c\in\C^\times$, then the eigenvalues of $\mathscr H(x)$ on $v_1$ and $v_2$ are different.

The conjecture follows from Theorem \ref{thm tensor irr} proved in Section \ref{sec third}.

The conjecture is clear for the case when $\varphi_{\bLa,\bm b}$ only has simple roots which we also recall from \cite{MVY15}. Note that $\dim L_{\bLa}(\bm b)=2^k$ and $\dim (L_{\bLa}(\bm b))^\sing=2^{k-1}$. If the polynomial $\varphi_{\bLa,\bm b}$ has no multiple roots, then $\varphi_{\bLa,\bm b}$ has the desired number of distinct monic divisors. Therefore, we have desired number of on-shell Bethe vectors which are also nonzero by Proposition \ref{prop bv nonzero}. By Theorem \ref{thm gl11 eigenvalue in y}, it implies that we do have an eigenbasis of the Gaudin transfer matrix consisting of on-shell Bethe vectors in $(L_{\bLa}(\bm b))^\sing$, see Proposition \ref{prop bv sing}, with different eigenvalues. Thus the algebraic Bethe ansatz works well for this situation. 

\begin{thm}[\cite{MVY15}]\label{thm complete}
Suppose all weights $\bla^{(s)}$, $1\lle s\lle k$, are polynomial $\g$-weights. If the polynomial $\varphi_{\bLa,\bm b}$ has no multiple roots, then the Gaudin transfer matrix $\mathscr H(x)$ is diagonalizable and the Bethe ansatz is complete. In particular, for any given $\bLa$ and generic $\bm b$, the Gaudin transfer matrix $\mathscr H(x)$ is diagonalizable and the Bethe ansatz is complete. \qed
\end{thm}

\section{Space $\mc{V}^\fkS$ and Weyl modules}\label{sec space VS}
In this section, we discuss the super-analog of $\mc{V}^S$ in \cite[Section 2.5]{MTV09}, cf. \cite[Section 3]{LM21}. 

The symmetric group $\fkS_n$ acts naturally on $\C[z_1,\dots,z_n]$ by permuting variables. Denote by $\sigma_i(\bm z)$ the $i$-th elementary symmetric polynomial in $z_1,\dots,z_n$. The algebra of symmetric polynomials $\C[z_1,\dots,z_n]^\fkS$ is freely generated by $\sigma_1(\bm z),\dots,\sigma_n(\bm z)$. 

Fix $\ell\in \{0,1,\dots,n\}$. We have a subgroup $\fkS_{\ell}\times \fkS_{n-\ell}\subset \fkS_n$.
Then $\fkS_{\ell}$ permutes the first $\ell$ variables while $\fkS_{n-\ell}$ permutes the last $n-\ell$ variables. Denote by $\C[z_1,\dots,z_n]^{\fkS_{\ell}\times \fkS_{n-\ell}}$ the subalgebra of $\C[z_1,\dots,z_n]$ consisting of $\fkS_{\ell}\times \fkS_{n-\ell}$-invariant polynomials. It is known that $\C[z_1,\dots,z_n]^{\fkS_{\ell}\times \fkS_{n-\ell}}$ is a free $\C[z_1,\dots,z_n]^\fkS$-module of rank $n\choose{\ell}$.

\subsection{Definition of $\mc{V}^\fkS$}
Let $V=(\C^{1|1})^{\otimes n}$ be the tensor power of the vector representation of $\g$. The $\g$-module $V$ has weight decomposition
\[
V=\bigoplus_{\ell=0}^{n}(V)_{(n-\ell,\ell)}.
\]
Let $\mc V$ be the space of polynomials in variables $\bm z=(z_1,z_2,\dots,z_n)$ with coefficients in $V$,
\[
\mc V=V\otimes \C[z_1,z_2,\dots,z_n].
\]The space $V$ is identified with the subspace $V\otimes 1$ of constant polynomials in $\mc V$. The space $\mc V$ has a natural grading induced from the grading on $\C[z_1,\dots,z_n]$ with $\deg(z_i)=1$. Namely, the degree of an element $v\otimes p$ in $\mc V$ is given by the degree of the polynomial $p$, $\deg(v\otimes p)=\deg\,p$. Clearly, the space $\End(\mc V)$ has a gradation structure induced from that on $\mc V$.

Let $P^{(i,j)}$ be the graded flip operator which acts on the $i$-th and $j$-th factors of $V$. Let $s_1$, $s_2$, $\dots$, $s_{n-1}$ be the simple permutations of the symmetric group $\fkS_n$. Define the $\fkS_n$-action on $\mc V$ by the rule:
\[
s_i:\bm f(z_1,\dots,z_n)\mapsto P^{(i,i+1)}\bm f(z_1,\dots,z_{i+1},z_i,\dots,z_n),
\]
for $\bm f(z_1,\dots,z_n)\in\mc V$. Note that the $\fkS_n$-action respects the gradation on $\mc V$. Denote the subspace of all vectors in $\mc V$ invariant with respect to the  $\fkS_n$-action by $\mc V^\fkS$. 

Clearly, the $\g$-action on $\mc V$ commutes with the $\fkS_n$-action on $\mc V$ and preserves the grading. Therefore, $\mc V^\fkS$ is a graded $\g$-module. Hence we have the weight decomposition for both $\mc V^\fkS$ and $(\mc V^\fkS)^\sing$,
\[
\mc V^\fkS=\bigoplus_{\ell=0}^{n}(\mc V^\fkS)_{(n-\ell,\ell)},\qquad (\mc V^\fkS)^\sing=\bigoplus_{\ell=0}^{n}(\mc V^\fkS)^\sing_{(n-\ell,\ell)}.
\]
Note that $(\mc V^\fkS)_{(n-\ell,\ell)}$ and $(\mc V^\fkS)^\sing_{(n-\ell,\ell)}$ are also graded $\C[z_1,\dots,z_n]^\fkS$-modules.

The space $\mc V$ is a $\g[t]$-module where $e_{ij}[r]$ acts by
\beq\label{eq:glt-action}
\begin{split}
e_{ij}[r](&p(z_1,\dots,z_n)w_1\otimes\dots\otimes w_n)\\=\ &p(z_1,\dots,z_n)\sum_{s=1}^n(-1)^{(|w_1|+\cdots+|w_{s-1}|)(|i|+|j|)}z_s^r\, w_1\otimes\dots\otimes e_{ij}w_s\otimes\dots\otimes w_n,
\end{split}
\eeq
for $p(z_1,\dots,z_n)\in\C[z_1,\dots,z_n]$ and $w_s\in \C^{1|1}$.

The following lemma is straightforward.
\begin{lem}\label{lem:graded-com}
The $\g[t]$-action on $\mc V$ commutes with the $\fkS_n$-action on $\mc V$. Both $\mc V$ and $\mc V^\fkS$ are graded $\g[t]$-modules.\qed
\end{lem}

\subsection{Properties of $\mc{V}^\fkS$ and $(\mc V^\fkS)^\sing$}
In this section, we recall properties of $\mc{V}^\fkS$ and $(\mc V^\fkS)^\sing$ from \cite[Section 3]{LM21}.
\begin{lem}\label{lem free graded}
The space $(\mc V^\fkS)_{(n-\ell,\ell)}$ is a free $\C[z_1,\dots,z_n]^\fkS$-module of rank $n\choose{\ell}$. In particular, the space $\mc V^\fkS$ is a free $\C[z_1,\dots,z_n]^\fkS$-module of rank $2^n$.\qed
\end{lem}

Set $v^+=v_1^{\otimes n}=v_1\otimes\dots\otimes v_1$.

\begin{lem}\label{lem cyclic sym graded}
The $\g[t]$-module $\mc V^\fkS$ is a cyclic module generated by $v^+$.\qed
\end{lem}

\begin{lem}\label{lem explicit basis} The set 
\beq\label{eq explicit basis}
\{e_{21}[r_1]e_{21}[r_2]\cdots e_{21}[r_\ell]v^+~|~0\lle r_1<r_2<\dots<r_{\ell}\lle n-1\}
\eeq
is a free generating set of $(\mc V^\fkS)_{(n-\ell,\ell)}$ over $\C[z_1,\dots,z_n]^\fkS$.\qed
\end{lem}
\begin{lem}\label{lem free sing graded}
The space $(\mc V^\fkS)^\sing_{(n-\ell,\ell)}$ is a free $\C[z_1,\dots,z_n]^\fkS$-module of rank ${n-1}\choose{\ell}$ with a free generating set given by
\beq\label{eq basis}
\{e_{12}[0]e_{21}[0]e_{21}[r_1]\cdots e_{21}[r_{\ell}]v^+,\quad 1\lle r_1<r_2<\dots<r_{\ell}\lle n-1\}.
\eeq
In particular, the space $(\mc V^\fkS)^\sing$ is a free $\C[z_1,\dots,z_n]^\fkS$-module of rank $2^{n-1}$.\qed
\end{lem}

Set $(q)_r=\prod_{i=1}^r(1-q^i)$.

\begin{prop}\label{prop ch}
We have
\[
\ch\big((\mc V^\fkS)_{(n-\ell,\ell)}\big)=\frac{q^{\ell(\ell-1)/2}}{(q)_{\ell}(q)_{n-\ell}},\qquad \ch\big( (\mc V^\fkS)^\sing_{(n-\ell,\ell)} \big)=\frac{q^{\ell(\ell+1)/2}}{(q)_\ell(q)_{n-1-\ell}(1-q^n)}.
\]
\end{prop}

Given $\bm a=(a_1,\dots,a_n)\in\C^n$, let $I_{\bm a}$ be the ideal of $\C[z_1,\dots,z_n]^\fkS$ generated by $\sigma_i(\bm z)- \bs a $, $i=1,\dots,n$. Then for any $\bm a$, by Lemmas \ref{lem:graded-com} and \ref{lem free graded}, the quotient space $\mc V^\fkS/I_{\bm a}\mc V^\fkS$ is a $\g[t]$-module of dimension $2^n$ over $\C$. Denote by $\bar v^+$ be the image of $v^+$ under this quotient.

\subsection{Weyl modules}
In this section, we discuss a special family of Weyl modules. We are not ware of references about Weyl modules of $\ugt$, cf. \cite{CLS19}, except the quantum affine case \cite{Zha17}.

Let $\eta(x)$ be a monic polynomial of degree $m$, where $m\in\Z_{\gge 0}$,
\[
\eta(x)=\sum_{i=0}^{m}\gamma_ix^i,\qquad \gamma_m=1.
\]
Denote by $W_{\eta}$ the $\g[t]$-module generated by an even vector $w$ subject to the relations: 
\beq\label{eq:weyl1}
e_{11}(x)w=\eta'(x)/\eta(x)w, \qquad e_{22}(x)w=e_{12}(x)w=0,
\eeq
\beq\label{eq:weyl2}
\sum_{i=0}^{m}\gamma_ie_{21}[i]w=0.
\eeq
It is convenient to write \eqref{eq:weyl2} as $(e_{21}\otimes \eta(t))w=0$. 

Clearly, we have $\dim W_\eta \lle 2^m$ by PBW theorem and \eqref{eq:anti}, \eqref{eq:weyl2}. The module $W_\eta$ is the universal $\g[t]$-module satisfying \eqref{eq:weyl1}, \eqref{eq:weyl2} which we call a {\it Weyl module}. 

If $\eta(x)=(x-b)^m$, we write $W_\eta$ as $W_m(b)$.

\begin{lem}\label{lem:weyl-0}
Let $\bm a=(0,\dots,0)\in \C^n$. Then $\mc V^\fkS/I_{\bm a}\mc V^\fkS$ is isomorphic to $W_n(0)$ as $\g[t]$-modules.
\end{lem}
\begin{proof}
It is clear that $\bar v^+$ satisfies the relations \eqref{eq:weyl1} for $\eta(x)=x^n$. In addition, it follows from the equation
\[
e_{21}[n]v^+=\sum_{i=1}^{n}(-1)^{i-1}\sigma_i(\bm z)e_{21}[n-i]v^+
\]
that $\bar v^+$ also satisfies the relation \eqref{eq:weyl2} for $\eta(x)=x^n$. Therefore, we have a surjective $\g[t]$-module homomorphism $W_{n}(0)\twoheadrightarrow \mc V^\fkS/I_{\bm a}\mc V^\fkS$. By $\dim \mc V^\fkS/I_{\bm a}\mc V^\fkS=2^n\gge \dim W_n(0)$, we obtain that it is also an isomorphism.
\end{proof}

In particular, we have $\dim W_m(b)=2^m$.

\begin{lem}\label{lem:weyl-tensor}
Let $\eta(x)=\prod_{s=1}^k(x-b_s)^{n_s}$, where $b_s\ne b_r$ for $1\lle s\ne r\lle k$. Then $W_\eta$ is isomorphic to $\bigotimes_{s=1}^k W_{n_s}(b_s)$ as $\g[t]$-modules.
\end{lem}
\begin{proof}
For $1\lle s\lle k$, let $w_s$ be the highest weight vector of $W_{n_s}(b_s)$ and set $\eta_s(x)=(x-b_s)^{n_s}$. Similar to Lemma \ref{lem:weyl-0}, since $\dim W_{\eta}\lle \dim \bigotimes_{s=1}^k W_{n_s}(b_s)$, it suffices to show that $\bigotimes_{s=1}^k W_{n_s}(b_s)$ is generated by $w^+:=\bigotimes_{s=1}^k w_s$ and $w^+$ satisfies the relations \eqref{eq:weyl1}, \eqref{eq:weyl2}. 

The proof of the fact that $\bigotimes_{s=1}^k W_{n_s}(b_s)$ is generated by $w^+$ is similar to the even case, see e.g. \cite{CP01}. The relations \eqref{eq:weyl1} are obvious. By $[e_{11}[r],e_{21}[i]]=-e_{21}[r+i]$ and relations \eqref{eq:weyl1}, \eqref{eq:weyl2}, we have 
$$
(e_{21}\otimes \eta_s(t)t^r)w_s=0,\qquad r\gge 0,\qquad 1\lle s\lle k.
$$
Therefore, we also have $(e_{21}\otimes \eta(t))w_s=0$ for all $1\lle s\lle k$. Now the relation \eqref{eq:weyl2} follows immediately.
\end{proof}

Given sequences $\bm n=(n_1,\dots,n_k)$ of nonnegative integers and $\bm b=(b_1,\dots,b_s)$ of distinct complex numbers, by Lemma \ref{lem:weyl-tensor}, we call $\bigotimes_{s=1}^k W_{n_s}(b_s)$ the {\it Weyl module associated with $\bm n$ and $\bm b$}.

\medskip

Given $\bm a=(a_1,\dots,a_n)\in\C^n$, define $k\in\Z_{>0}$, $b_s\in \C$, and $n_s\in \Z_{>0}$ for $1\lle s\lle k$ by
\beq\label{eq:a-b-relations}
x^n+\sum_{i=1}^n (-1)^ia_ix^{n-i}=\prod_{s=1}^k(x-b_s)^{n_s}.
\eeq
Note that $n=\sum_{s=1}^k n_s$.
\begin{lem}\label{lem local weyl}
The $\g[t]$-module $\mc V^\fkS/ I_{\bm a}\mc V^\fkS$ is isomorphic to the Weyl module $\bigotimes_{s=1}^k W_{n_s}(b_s)$.
\end{lem}
\begin{proof}
Similar to the proof of Lemma \ref{lem:weyl-0}, the statement follows from Lemma \ref{lem:weyl-tensor} by checking relations and comparing dimensions.
\end{proof}

We also need the following statements.
\begin{lem}\label{lem:weyl-0-more}
We have the following properties for $W_m(0)$.
\begin{enumerate}
    \item The module $W_m(0)$ has a unique grading such that $W_m(0)$ is a graded $\g[t]$-module and such that the degree of the highest weight vector $w$ equals $0$.
    \item As a $\g$-module, $W_m(0)$ is isomorphic to $(\C^{1|1})^{\otimes m}$.
    \item A $\g[t]$-module $M$ is an irreducible subquotient of $W_m(0)$ if and only if $M$ has the form $L_{\bla}(0)$, where $\bla$ is a polynomial weight such that $|\bla|=m$.
\end{enumerate}
\end{lem}
\begin{proof}
The first two statements follows from Lemma \ref{lem:weyl-0} and the construction of $\mc V^{\fkS}/I_{\bm a}\mc V^{\fkS}$. The last statement follows from the first two.
\end{proof}

\begin{lem}\label{lem:weyl-b-more}
Let $b\in \C$. We have the following properties for $W_m(b)$.
\begin{enumerate}
    \item As a $\g$-module, $W_m(b)$ is isomorphic to $(\C^{1|1})^{\otimes m}$.
    \item A $\g[t]$-module $M$ is an irreducible subquotient of $W_m(b)$ if and only if $M$ has the form $L_{\bla}(b)$, where $\bla$ is a polynomial weight such that $|\bla|=m$.
\end{enumerate}
\end{lem}
\begin{proof}
It follows from Lemma \ref{lem:weyl-0-more}.
\end{proof}

\begin{cor}
A $\g[t]$-module $M$ is an irreducible subquotient of $\bigotimes_{s=1}^k W_{n_s}(b_s)$ if and only if $M$ has the form $\bigotimes_{s=1}^k L_{\bla^{(s)}}(b_s)$, where $\bla^{(s)}$ is a polynomial weight such that $|\bla^{(s)}|=n_s$ for each $1\lle s\lle k$.
\end{cor}
\begin{proof}
It follows from part (ii) of Lemma \ref{lem:weyl-b-more}  and the irreducibility of $\bigotimes_{s=1}^k L_{\bla^{(s)}}(b_s)$, and the Jordan-H\"{o}lder theorem.
\end{proof}

\section{Main theorems}\label{sec main thms}
\subsection{The algebra $\mc{O}_{l}$}\label{sec o_l}
Let $\Omega_{l}$ be the $n$-dimensional affine space with coordinates $f_1,\dots,f_l$, $g_1$, $\dots$, $g_{n-l-1}$ and $\vSi_n$. Introduce two polynomials 
\beq\label{eq poly f g}
f(x)=x^l+\sum_{i=1}^{l}f_ix^{l-i},\quad g(x)=x^{n-l-1}+\sum_{i=1}^{n-l-1}g_ix^{n-l-i-1}.
\eeq
Denote by $\mc O_l$ the algebra of regular functions on $\Omega_{l}$, namely $$\mc O_l=\C[f_1,\dots,f_l,g_1,\dots,g_{n-l-1},\vSi_n].$$ Define the degree function by
\[
\deg f_i=i,\qquad \deg g_j=j,\qquad \deg\vSi_n=n,
\]for all $i=1,\dots,l$ and $j=1,\dots,n-l-1$. The algebra $\mc O_l$ is graded with the graded character given by
\beq\label{eq ch O}
\ch(\mc O_l)=\frac{1}{(q)_l(q)_{n-l-1}(1-q^n)}.
\eeq

Let $\varSigma_1,\dots,\varSigma_{n-1}$ be the elements of $\mc O_l$ such that
\beq\label{eq si q=1}
nf(x)g(x)=nx^{n-1}+\sum_{i=1}^{n-1}(-1)^i(n-i)\varSigma_i x^{n-1-i}.
\eeq
The homomorphism
\beq\label{eq inj sym}
\pi_l:\C[z_1,\dots,z_n]^\fkS\to \mc O_l,\qquad \sigma_i(\bm z)\mapsto \vSi_i,\qquad i=1,\dots,n,
\eeq
is injective and induces a $\C[z_1,\dots,z_n]^\fkS$-module structure on $\mc O_l$.

Express $f'(x)g(x)$ as follows,
\beq\label{eq G q=1}
nf'(x)g(x)=nlx^{n-2}+\sum_{i=1}^{n-2}G_{i}x^{n-2-i},
\eeq
where $G_{i}\in\mc O_l$. 

\begin{lem}
The elements $G_{i}$ and $\varSigma_j$, $i=1,\dots,n-2$, $j=1,\dots, n$, generate the algebra $\mc O_l$.\qed
\end{lem}

\begin{lem}\label{lem deg O}
The elements $G_{i}$ and $\varSigma_j$ are homogeneous of degrees $i$ and $j$, respectively, for $i=1,\dots,n-2$, $j=1,\dots, n$.\qed
\end{lem}

\subsection{Bethe algebra}\label{sec ber}
We call the unital subalgebra of $\ugt$ generated by the coefficients of 
\[
e_{11}(x)+e_{22}(x),\quad 2\mathscr H(x)=e_{11}(x)e_{11}(x)-e_{12}(x)e_{21}(x)+e_{21}(x)e_{12}(x)-e_{22}(x)e_{22}(x)
\]
the \emph{Bethe algebra}. We denote the Bethe algebra by $\mc B$. Note that the coefficients of $e_{11}(x)+e_{22}(x)$ generate the center of $\ugt$.

\begin{lem}[{\cite{MR14}}]
The Bethe algebra $\mc B$ is commutative. The Bethe algebra $\mc B$ commutes with the subalgebra $\mathrm{U}(\g)\subset \ugt$.\qed
\end{lem}

Being a subalgebra of $\ugt$, the Bethe algebra $\mc B$ acts on any $\g[t]$-module $M$. Since $\mc B$ commutes with $\mathrm{U}(\g)$, the Bethe algebra preserves the subspace $(M)_{\bla}^\sing$ for any weight $\bla$. If $K\subset M$ is a $\mc B$-invariant subspace, then we call the image of $\mc B$ in $\End(K)$ the \textit{Bethe algebra associated with} $K$.

Let $\bm a=(a_1,\dots,a_n)\in\C^n$. Define $k\in \Z_{>0}$, a sequence of positive integers $\bm n=(n_1,\dots,n_k)$, and a sequence of distinct complex numbers $\bm b=(b_1,\dots,b_k)$ by \eqref{eq:a-b-relations}. Let $\bLa=(\bla^{(1)},\dots,\bla^{(k)})$ be a sequence of polynomial $\g$-weights such that $|\bla^{(s)}|=n_s$. 

We study the action of the Bethe algebra $\mc B$ on the following $\mc B$-modules:
\[
\mc M_{l}=(\mc V^\fkS)_{(n-l,l)}^\sing, \quad
\mc M_{l,\bm a}=(\bigotimes_{s=1}^k W_{n_s}(b_s))_{(n-l,l)}^\sing,\quad 
\mc M_{l,\bLa,\bm b}=(\bigotimes_{s=1}^k L_{\bla^{(s)}}(b_s))_{(n-l,l)}^\sing.
\]
Denote the Bethe algebras associated with $\mc M_{l}$, $\mc M_{l,\bm a}$, $\mc M_{l,\bLa,\bm b}$ by $\mc B_{l}$, $\mc B_{l,\bm a}$, $\mc B_{l,\bLa,\bm b}$, respectively. For any element $X\in \mc B$, we denote by $X(\bm z)$, $X(\bm a)$, $X(\bLa,\bm b)$ the respective linear operators.

Since by Lemma \ref{lem cyclic sym graded} the $\g[t]$-module $\mc V^\fkS$ is generated by $v_1^{\otimes n}=v_1\otimes\dots\otimes v_1$, the series $e_{11}(x)+e_{22}(x)$ acts on $\mc V^\fkS$ by multiplication by the series
\[
\sum_{i=1}^n \frac{1}{x-z_i}=\sum_{i=1}^n \sum_{j=0}^\infty z_i^j x^{-j-1}.
\]
Therefore there exist unique central elements $C_1,\dots,C_n$ of $\ugt$ of minimal degrees such that each $C_i$ acts on $\mc V^\fkS$ by multiplication by $\sigma_i(\bm z)$. 

Define $B_{i} \in \mc B $ by
\beq\label{eq diff poly}
\Big(x^n+\sum_{i=1}^n(-1)^iC_ix^{n-i}\Big)\mc T(x)=x^{n} \sum_{i=2}^{\infty} B_ix^{-i},
\eeq
where $\mc T(x)$ is defined in \eqref{eq:T-transfer}.

\begin{lem}\label{lem:B-vanish}
We have $B_i(\bm z)=0$ for $i>n$ and $B_2(\bm z)=nl$.
\end{lem}
\begin{proof}
Let $V(\bm c)=\bigotimes_{i=1}^n \C^{1|1}(c_i)$, where $c_i\in \C$. Note that $B_i(\bm z)$ is a polynomial in $\bm z$ with values in $\End((V)_{(n-l,l)}^\sing)$. For any sequence of complex numbers $\bm c=(c_1,\dots,c_n)$, we can evaluate $B_i(\bm z)$ at $\bm z=\bm c$ to an operator on $(V(\bm c))_{(n-l,l)}^\sing$. By Theorem \ref{thm complete}, the Gaudin transfer matrix $\mathscr H(x)$ is diagonalizable and the Bethe ansatz is complete for $(V(\bm c))_{(n-l,l)}^\sing$ when $\bm c\in\C^n$ is generic. Hence by \eqref{eq:T-eigenvalue} and \eqref{eq diff poly} that $\big(x^n+\sum_{i=1}^n(-1)^iC_ix^{n-i}\big)\mc T(x)$ acts on $(V(\bm c))_{(n-l,l)}^\sing$ as a polynomial in $x$ for generic $\bm c$. In particular, it implies that $B_i$, $i>n$, acts on $(V(\bm c))_{(n-l,l)}^\sing$ by zero for generic $\bm c$.  Therefore $B_i(\bm z)$, $i>n$, is identically zero.

By the same reasoning, one shows that $B_2(\bm z)=nl$. Alternatively, it also follows from $B_2=(e_{11}+e_{22})e_{22}-e_{21}e_{12}$.
\end{proof}

\begin{lem}\label{lem gen B}
The elements $B_{i} (\bm z)$ and $C_j(\bm z)$, $i=3,\dots, n$, $j=1,\dots,n$, generate the algebra $\mc B_l$.
\end{lem}
\begin{proof}
It follows from the definition of $\mc B$, \eqref{eq diff poly}, and Lemma \ref{lem:B-vanish}.
\end{proof}

\begin{lem}\label{lem deg B}
The elements $B_i$ and $C_{j}$ are homogeneous of degrees $i-2$ and $j$, respectively, for $i=3,4,\dots$ and $j=1,\dots,n$.\qed
\end{lem}

\subsection{Main theorems}\label{sec main thm}
Recall from Proposition \ref{prop ch} that there exists a unique vector (up to proportionality) of degree $l(l+1)/2$ in $\mc M_{l}$ explicitly given by $$\mathfrak u_l:=e_{12}[0]e_{21}[0]e_{21}[1]\cdots e_{21}[l]v^+,$$ see Lemma \ref{lem free sing graded}.

Any commutative algebra $\mc A$ is a module over itself induced by left multiplication. We call it the \emph{regular representation of} $\mc A$. The dual space 
$\mc A^*$ is naturally an $\mc A$-module which is called the \emph{coregular representation}. A bilinear form $(\cdot|\cdot):\mc A\otimes \mc A\to \C$ is called \emph{invariant} if $(ab|c)=(a|bc)$ for all $a,b,c\in\mc A$. A finite-dimensional commutative algebra $\mc A$ admitting an invariant non-degenerate symmetric bilinear form $(\cdot|\cdot):\mc A\otimes \mc A\to \C$ is called a \emph{Frobenius algebra}. The regular and coregular representations of a  Frobenius algebra are isomorphic.

Let $M$ be an $\mathcal A$-module and $\mathcal E:\mathcal A\to \C$ a character, then the {\it $\mc A$-eigenspace associated to $\mc E$} in $M$ is defined by $\bigcap_{a\in \mathcal A}\ker(a|_M-\mathcal E(a))$. The {\it generalized $\mc A$-eigenspace associated to $\mc E$} in $M$ is defined by $\bigcap_{a\in \mathcal A}\big(\bigcup_{m=1}^\infty\ker(a|_M-\mathcal E(a))^m\big)$.

\begin{thm}\label{thm VS}
The action of the  Bethe algebra $\mc B_l$ on $\mc M_{l}$ has the following properties.
\begin{enumerate}
\item The map $\eta_l:G_{i}\mapsto B_{i+2}(\bm z)$, $\varSigma_j\mapsto C_j(\bm z)$, $i=1,\dots,n-2$, $j=1,\dots,n$, extends uniquely to an isomorphism $\eta_l:\mc O_l\to \mc B_l$ of graded algebras. Moreover, the isomorphism $\eta_l$ is an isomorphism of $\C[z_1,\dots,z_n]^\fkS$-modules.
	
\item The map $\rho_l:\mc O_l\mapsto \mc M_{l}$, $F\mapsto \eta_l(F)\mathfrak u_l$, is an isomorphism of graded vector spaces identifying the $\mc B_l $-module $\mc M_{l}$ with the regular representation of $\mc O_l$.
\end{enumerate}
\end{thm}

Theorem \ref{thm VS} is proved in Section \ref{sec proof}.

\medskip

Let $\bm a=(a_1,\dots,a_n)\in\C^n$. Define $k\in \Z_{>0}$, a sequence of positive integers $\bm n=(n_1,\dots,n_k)$, and a sequence of distinct complex numbers $\bm b=(b_1,\dots,b_k)$ by \eqref{eq:a-b-relations}. Let $\bLa=(\bla^{(1)},\dots,\bla^{(k)})$ be a sequence of non-degenerate polynomial weights such that $|\bla^{(s)}|=n_s$ for each $1\lle s\lle k$.

\begin{thm}\label{thm tensor irr}
The action of the Bethe algebra $\mc B_{l,\bLa,\bm b}$ on $\mc M_{l,\bLa,\bm b}$ has the following properties.
\begin{enumerate}
\item The Bethe algebra $\mc B_{l,\bLa,\bm b}$ is isomorphic to
\[
\C[w_1,\dots,w_{k-1}]^{\fkS_{l}\times \fkS_{k-l-1}}/\langle \sigma_i(\bm w)-\varepsilon_i\rangle_{i=1,\dots,k-1},\]
where $\varepsilon_i$ are given by 
$$
\varphi_{\bLa,\bm b}(x):=\prod_{s=1}^k(x-b_s)\sum_{s=1}^k\frac{n_s}{x-b_s}=n\Big(x^{k-1}+\sum_{i=1}^{k-1}(-1)^i\varepsilon_ix^{k-1-i}\Big)$$ and $\sigma_i(\bm w)$ are elementary symmetric functions in $w_1,\dots,w_{k-1}$. 
\item The Bethe algebra $\mc B_{l,\bLa,\bm b}$ is a Frobenius algebra. Moreover, the $\mc B_{l,\bLa,\bm b}$-module $\mc M_{l,\bLa,\bm b}$ is isomorphic to the regular representation of $\mc B_{l,\bLa,\bm b}$.
\item The Bethe algebra $\mc B_{l,\bLa,\bm b}$ is a maximal commutative subalgebra in $\mc M_{l,\bLa,\bm b}$ of dimension $\binom{k-1}{l}$.
\item Every $\mc B$-eigenspace in $\mc M_{l,\bLa,\bm b}$ has dimension one.
\item The $\mc B$-eigenspaces in $\mc M_{l,\bLa,\bm b}$ bijectively correspond to the monic degree $l$ divisors $y(x)$ of the polynomial $\varphi_{\bLa,\bm b}(x)$. Moreover, the eigenvalue of $\mathscr H(x)$ corresponding to the monic divisor $y$ is described by $\mc E_{y,\bLa,\bm b}(x)$, see \eqref{eq gl11 eigenvalue in y}.
\item Every generalized $\mc B$-eigenspace in $\mc M_{l,\bLa,\bm b}$ is a cyclic $\mc B$-module.
\item The dimension of the generalized $\mc B$-eigenspace associated to $\mc E_{y,\bLa,\bm b}(x)$ is 
\[
\prod_{a\in \C}\binom{\mathrm{Mult}_a(\varphi_{\bLa,\bm b})}{\mathrm{Mult}_a(y)},
\]where $\mathrm{Mult}_a(p)$ is the multiplicity of $a$ as a root of the polynomial $p$.
\end{enumerate}
\end{thm}

Theorem \ref{thm tensor irr} is proved in Section \ref{sec proof}. 

Note that the results of is quite parallel to that of XXX spin chains, see \cite[Theorem 4.11]{LM21}.

\subsection{Higher Gaudin transfer matrices}

To define higher Gaudin transfer matrices, we first recall basics about pseudo-differential operators. Let $\mathscr A$ be a differential superalgebra with an even derivation $\pa:\mathscr A\to\mathscr A$. For $r\in\Z_{>0}$, denote the $r$-th derivative of $a\in\mathscr A$ by $a_{[r]}$. Define the \emph{superalgebra of pseudo-differential operators} $\mathscr A((\pa^{-1}))$ as follows. Elements of $\mathscr A((\pa^{-1}))$ are Laurent series in $\pa^{-1}$ with coefficients in $\mathscr A$, and the product is given by
\[
\pa\pa^{-1}=\pa^{-1}\pa=1,\quad \pa^r a=\sum_{s=0}^\infty {r \choose s}a_{[s]}\pa^{r-s},\quad r\in\Z,\quad a\in\mathscr A,
\]
where 
$$
{r \choose s}=\frac{r(r-1)\cdots(r-s+1)}{s!}.
$$

Let 
$$
\mathscr A_{x}^{m|n}=\mathrm{U}(\g[t])((x^{-1}))=\Big\{ \sum_{r=-\infty}^s g_rx^{r},\ r\in \Z,\ g_r\in \mathrm{U}(\g[t]) \Big\}
$$
Consider the operator in $\End(\C^{1|1})\otimes \mathscr A_{x}^{m|n}((\pa_x^{-1}))$,
$$
\mathfrak Z(x,\pa_x):=\sum_{a,b=1}^2E_{ab}\otimes\left(\delta_{ab}\pa_x-e_{ab}(x)(-1)^{|a|}\right).
$$
which is a Manin matrix, see \cite[Lemma 3.1]{MR14} and \cite[Lemma 4.2]{HM20}. Define the {\it Berezinian}, see \cite{Naz}, of $\mathfrak Z(x,\pa_x)$ by
\beq\label{eq diff trans}
\mathrm{Ber}\big(\mathfrak Z(x,\pa_x)\big)=\big(\pa_x-e_{11}(x)\big)\Big(\pa_x+e_{22}(x)+e_{21}(x)\big(\pa_x-e_{11}(x)\big)^{-1}e_{12}(x)\Big)^{-1}.
\eeq
Denote the Berezinian by $\mathfrak{D}(x,\pa_x)$ and expand it as an element in $\mathscr A_{x}^{m|n}((\pa_x^{-1}))$,
\beq\label{eq:Ber-Gaudin}
\mathfrak{D}(x,\pa_x)=\sum_{r=0}^{\infty}(-1)^r\mathcal G_{r}(x)\pa_x^{-r}.
\eeq
We call the series $\mathcal G_{r}(x)\in \mathscr A_{x}^{m|n}$, $r\in \Z_{\gge 0}$, the \emph{higher Gaudin transfer matrices}. In particular, we call $\mathcal G_1(x)$ and $\mathcal G_2(x)$ the {\it first and second Gaudin transfer matrices}, respectively.

\begin{eg}\label{eg:transfer}
We have $\mathcal G_0(x)=1$,
\[
\mathcal G_1(x)=e_{11}(x)+e_{22}(x),\quad \mathcal G_2(x)=\big(e_{11}(x)+e_{22}(x)\big)e_{22}(x)-e_{21}(x)e_{12}(x).
\]
Moreover, we have
\[
\mathscr H(x)=\frac{1}{2}\big(\mathcal G_1(x)\big)^2-\mathcal G_2(x)+\frac{1}{2}\pa_x\mathcal G_1(x),\quad \mc T(x)=\mathcal G_2(x),
\]
see \eqref{eq:T-transfer}.
\qed
\end{eg}

\begin{rem}
In principal, the Bethe algebra should be the unital subalgebra of $\ugt$ generated by coefficients $\mathcal G_{r}(x)$, $r\in \Z_{> 0}$, cf. \cite{MM15}. However, it turns out that the first two transfer matrices already give (almost) complete information about the Bethe algebra, see the discussion below.\qed
\end{rem}

Now we describe the eigenvalues of higher Gaudin transfer matrices acting on the on-shell Bethe vector.

Let $\bLa=(\bla^{(1)},\dots,\bla^{(k)})$ be a sequence of $\g$-weights and $\bm b=(b_1,\dots,b_k)$ a sequence of distinct complex numbers, where $\bla^{(s)}=(\alpha_s,\beta_s)$. Let $\bm t=(t_1,\dots,t_l)$, where $0\lle l<k$. Suppose that $y_{\bm t}$ divides the polynomial $\varphi_{\bLa,\bm b}$ (namely $\bm t$ satisfies the Bethe ansatz equation), see \eqref{eq gl11 BAE}.
\begin{thm}\label{thm:D-eigenvalue-b}
If $t_i\ne t_j$ for $1\lle i< j\lle l$, then
\beq\label{eq:eigen-D}
\mathfrak{D}(x,\pa_x)\mathbb B_l(\bm t)=\mathbb B_l(\bm t)\Big(\pa_x-\sum_{s=1}^k\frac{\alpha_s}{x-b_s}+\frac{y_{\bm t}'}{y_{\bm t}}\Big)\Big(\pa_x+\sum_{s=1}^k\frac{\beta_s}{x-b_s}+\frac{y_{\bm t}'}{y_{\bm t}}\Big)^{-1}.
\eeq
\end{thm}
This is a differential analog of \cite[Theorem 6.4]{LM21}. Note that the pseudo-differential operator in the right hand side of \eqref{eq:eigen-D}, denoted by $\mathfrak D_{y,\bLa,\bm b}$, was introduced \cite[Section 5.3]{HMVY}. This theorem can be generalized to the $\glMN$ case where on the right hand side the pseudo-differential operator describing the eigenvalues of higher Gaudin transfer matrices should be replaced by the pseudo-differential operator in \cite[Equation (6.5)]{HMVY}. This generalization is a classical limit of \cite[Conjecture 5.15]{LM:2020} which connects the rational difference operator introduced in \cite[Equation (5.6)]{HLM} with the eigenvalues of higher transfer matrices on the on-shell Bethe vector for XXX spin chains associated with $\gl(m|n)$. 

The proof of Theorem \ref{thm:D-eigenvalue-b} can be obtained from \cite[Theorem 6.4]{LM21} by taking classical limit, see \cite{MTV06} and cf. \cite[Remark 5.4]{HLM}. We shall provide a proof of  \cite[Conjecture 5.15]{LM:2020} and the generalization of of Theorem \ref{thm:D-eigenvalue-b} to the $\glMN$ case in a further publication using nested Bethe ansatz, see \cite{KR83,MTV06,BR08}. Hence the proof of Theorem \ref{thm:D-eigenvalue-b} is omitted here.

\begin{rem}
As shown in \cite[Lemma 5.7]{HMVY}, the odd reflections of $\mathfrak D_{y,\bLa,\bm b}$ cf. \cite[equation (3.1)]{HMVY}, which come from the study of the fermionic reproduction procedure of the Bethe ansatz equation, are compatible with the odd reflections in Lie superalgebras. The difference analog of this fact has been used in \cite{Lu21} to investigate the odd reflections of super Yangian of type A and the fermionic reproduction procedure of the Bethe ansatz equation for XXX spin chains.\qed
\end{rem}

We conclude this section by discussing the connections between $\mathcal G_i(x)$, $i\gge 3$, and $\mathcal G_1(x)$, $\mathcal G_2(x)$.

Let
\[
\mu(x)=\sum_{s=1}^k\frac{\alpha_s}{x-b_s}-\frac{y_{\bm t}'}{y_{\bm t}},\quad  \nu(x)=\sum_{s=1}^k\frac{\beta_s}{x-b_s}+\frac{y_{\bm t}'}{y_{\bm t}}.
\]For simplicity, we do not write the dependence of $\mu(x)$ and $\nu(x)$ on $\bLa,\bm b,\bm t$ explicitly. Then the eigenvalue of $\mathfrak{D}(x,\pa_x)$ acting on $\mathbb B_l(\bm t)$ is given by
\beq\label{eq:D}
(\pa_x-\mu(x))(\pa_x+\nu(x))^{-1}=1-(\mu(x)+\nu(x))(\pa_x+\nu(x))^{-1}.
\eeq
Hence the eigenvalues of $\mathcal G_i(x)$ are essentially determined only by $\mu(x)+\nu(x)$ and $\nu(x)$. Comparing \eqref{eq:Ber-Gaudin} and the expansion of \eqref{eq:D}, we have
\beq\label{eq:D2}
\mathcal G_1(x)\mathbb B_l(\bm t)=(\mu(x)+\nu(x))\mathbb B_l(\bm t), \quad \mathcal G_2(x)\mathbb B_l(\bm t)=(\mu(x)+\nu(x))\nu(x)\mathbb B_l(\bm t), 
\eeq
see also \eqref{eq:T-eigenvalue}. Therefore, the spectrum of all higher transfer matrices are determined simply by that of the first two transfer matrices which justifies our definition of Bethe algebra.
\begin{lem}\label{lem:zero-inter}
Let the complex parameters $c_1,\dots,c_m$ and the positive integer $m$ vary. Then the kernels of the representations $\bigotimes_{i=1}^m\C^{1|1}(c_i)$ of $\ugt$ have the zero intersection.
\end{lem}
\begin{proof}
The proof is contained in the proof of \cite[Proposition 1.7]{Naz20}.
\end{proof}
\begin{cor}\label{universal oper}
We have
\beq\label{eq:D-alt}
\mathfrak{D}(x,\pa_x)=\Big(\pa_x-\mathcal G_1(x)+\frac{\mathcal G_2(x)}{\mathcal G_1(x)}\Big)\Big(\pa_x+\frac{\mathcal G_2(x)}{\mathcal G_1(x)}\Big)^{-1}.
\eeq
\end{cor}
\begin{proof}
By Lemma \ref{lem:zero-inter}, it suffices to check that the left hand side and the right hand side of \eqref{eq:D-alt} act identically on a basis of $\bigotimes_{i=1}^m\C^{1|1}(c_i)$ for all $m\in \Z_{>0}$ and generic $\bm c=(c_1,\dots,c_m)$. 

By Theorem \ref{thm complete}, there is a basis of $\bigotimes_{i=1}^m\C^{1|1}(c_i)$ consisting of on-shell Bethe vectors for generic $\bm c$. Therefore, the statement follows from Theorem \ref{thm:D-eigenvalue-b}, \eqref{eq:D}, \eqref{eq:D2}.
\end{proof}

\section{Proof of main theorems}\label{sec proof}
In this section, we prove the main theorems. 
\subsection{The first isomorphism}
\begin{proof}[Proof of Theorem \ref{thm VS}]
We first show the homomorphism defined by $\eta_l$ is well-defined. 

Consider the tensor product $V(\bs c)=\bigotimes_{i=1}^n\C^{1|1}(c_i)$, where $c_i\in \C$, and the corresponding Bethe ansatz equation associated to weight $(n-l,l)$.
Let $\bs t$ be a solution with distinct coordinates and ${\bB}_l(\bm t)$ the corresponding on-shell Bethe vector. Denote $\mc E_{i,\bs t}$ the eigenvalues of $B_i$ acting on ${\bB}_l(\bm t)$, see Theorem \ref{thm gl11 eigenvalue in y} and \eqref{eq:T-eigenvalue}. 

Define a character $\pi:\mc O_l\to \C$ by sending  
$$ f(x)\mapsto y_{\bs t}(x),\qquad g(x)\mapsto\frac1{ny_{\bs t}(x)}\ \prod_{i=1}^n(x-c_i)\sum_{i=1}^n\frac{1}{x-c_i},\qquad\varSigma_n\mapsto \prod_{i=1}^n c_i.$$

Then
\beq\label{equal on bethe vector} 
\pi(\varSigma_i)= \sigma_i(\bs c), \qquad \pi(G_i)=\mc E_{i,\bs t},
\eeq
by \eqref{eq si q=1} and by \eqref{eq gl11 eigenvalue in y}, \eqref{eq:T-eigenvalue}, \eqref{eq G q=1},  respectively.

\medskip

Let now $P(G_{i},\varSigma_j)$ be a polynomial in $G_{i},\varSigma_j$ such that $P(G_{i},\varSigma_j)$ is equal to zero in $\mc O_{l}$. It suffices to show $P(B_i(\bm z),C_j(\bm z))$ is equal to zero in $\mc B_l$. 

Note that $P(B_i(\bm z),C_j(\bm z))$ is a polynomial in $z_1,\dots,z_n$ with values in $\End((V)_{(n-l,l)}^\sing)$. For any sequence $\bm c$ of complex numbers, we can evaluate $P(B_i(\bm z),C_j(\bm z))$ at $\bs z=\bs c$ to an operator on $(V(\bs c))_{(n-l,l)}^\sing$. By Theorem \ref{thm complete}, the transfer matrix $\mc T(x)$ is diagonalizable and the Bethe ansatz is complete for $(V(\bm c))_{(n-l,l)}^\sing$ when $\bm c\in\C^n$ is generic. Hence by \eqref{equal on bethe vector} the value of $P(B_i(\bm z),C_j(\bm z))$ at $\bs z=\bs c$ is also equal to zero for generic $\bs c$. Therefore $P(B_i(\bm z),C_j(\bm z))$ is identically zero and the map $\eta_l$ is well-defined.

\medskip

Let us now show that the map $\eta_l$ is injective. Let $P(G_{i},\varSigma_j)$ be a polynomial in $G_{i},\varSigma_j$ such that $P(G_{i},\varSigma_j)$ is non-zero in $\mc O_{l}$. Then the value at a generic point of $\Omega_l$ (e.g. the non-vanishing points of $P(G_{i},\varSigma_j)$ such that $f$ and $g$ are relatively prime and have only simple zeros) is not equal to zero. Moreover, at those points the transfer matrix $\mc T(x)$ is diagonalizable and the Bethe ansatz is complete again by Theorem \ref{thm complete}. Therefore, again by \eqref{equal on bethe vector}, the polynomial $P(B_i(\bm z),C_j(\bm z))$ is a non-zero element in $\mc B_l$. Thus the map $\eta_l$ is injective.

The surjectivity of $\eta_l$ follows from Lemma \ref{lem gen B}. Hence $\eta_l$ is an isomorphism of algebras. 

The fact that $\eta_l$ is an isomorphism of graded algebra respecting the gradation follows from Lemmas \ref{lem deg O} and \ref{lem deg B}. This completes the proof of part (i).

\medskip

The kernel of $\rho_l$ is an ideal of $\mc O_l$. If we identify $\sigma_i(\bm z)$ with $\vSi_i$, then the algebra $\mc O_l$ contains the algebra $\C[z_1,\dots,z_n]^\fkS$, see \eqref{eq inj sym}. The kernel of $\rho_l$ intersects $\C[z_1,\dots,z_n]^\fkS$ trivially. Therefore the kernel of $\rho_l$ is trivial as well. Hence $\rho_l$ is an injective map. Comparing \eqref{eq ch O} and Proposition \ref{prop ch}, we have $\ch\big(\mc M_{l}\big)=q^{l(l+1)/2}\ch(\mc O_l)$. Thus $\rho_l$ is an isomorphism of graded vector spaces which shifts the degree by $l(l+1)/2$, completing the proof of part (ii).
\end{proof}

\subsection{The second isomorphism}\label{sec sec-iso}
Let $\bm a=(a_1,\dots,a_n)$ be a sequence of complex numbers. Define $k\in \Z_{>0}$, a sequence of positive integers $\bm n=(n_1,\dots,n_k)$, and a sequence of distinct complex numbers $\bm b=(b_1,\dots,b_k)$ by \eqref{eq:a-b-relations}. Let $I_{l,\bm a}^{\mc O}$ be the ideal of $\mc O_l$ generated by the elements $\vSi_i-a_i$, $i=1,\dots,n$, where $\vSi_1,\dots,\vSi_{n-1}$ are defined in
\eqref{eq si q=1}. Let $\mc O_{l,\bm a}$ be the quotient algebra
\[
\mc O_{l,\bm a}=\mc O_{l}/I_{l,\bm a}^{\mc O}.
\]

Let $I_{l,\bm a}^{\mc B}$ be the ideal of $\mc B_l $ generated by $C_i(\bm z)-a_i$, $i=1,\dots,n$. Consider the subspace 
\[
I_{l,\bm a}^{\mc M}=I_{l,\bm a}^{\mc B}\mc M_l=(I_{\bm a}\mc V^\fkS)_{(n-l,l)}^\sing,
\]
where $I_{\bm a}$ as before is the ideal of $\C[z_1,\dots,z_n]^\fkS$ generated by $\sigma_i(\bm z)-a_i$, $i=1,\dots,n$. 

\begin{lem}\label{lem eta rho weyl}
We have
\[
\eta_l(I_{l,\bm a}^{\mc O})=I_{l,\bm a}^{\mc B},\quad \rho_l(I_{l,\bm a}^{\mc O})=I_{l,\bm a}^{\mc M},\quad \mc B_{l,\bm a}=\mc B_{l}/I_{l,\bm a}^{\mc B},\quad \mc M_{l,\bm a}=(\mc V^\fkS)^\sing_{(n-l,l)}/ I_{l,\bm a}^{\mc M}.
\]
\end{lem}
\begin{proof}
The lemma follows from Theorem \ref{thm VS} and Lemma \ref{lem local weyl}.
\end{proof}

By Lemma \ref{lem eta rho weyl}, the maps $\eta_l$ and $\rho_l$ induce the maps
\[
\eta_{l,\bm a}:\mc O_{l,\bm a}\to \mc B_{l,\bm a},\qquad \rho_{l,\bm a}:\mc O_{l,\bm a}\to \mc M_{l,\bm a}.
\]The map $\eta_{l,\bm a}$ is an isomorphism of algebras. Since $\mc B_{l,\bm a}$ is finite-dimensional, by e.g. \cite[Lemma 3.9]{MTV09}, $\mc O_{l,\bm a}$ is a Frobenius algebra, so is 
$\mc B_{l,\bm a}$. The map $\rho_{l,\bm a}$ is an isomorphism of vector spaces. Moverover, it follows from Theorem \ref{thm VS} and Lemma \ref{lem eta rho weyl} that $\rho_{l,\bm a}$ identifies the regular representation of $\mc O_{l,\bm a}$ with the $\mc B_{l,\bm a}$-module $\mc M_{l,\bm a}$.

The statement of this section implies, by e.g. \cite[Lemma 1.3]{Lu20}, the following. Set 
$$
\zeta_{\bm n,\bm b}(x)=\sum_{s=1}^k\frac{n_s}{x-b_s},\qquad \psi_{\bm n,\bm b}(x):=\zeta_{\bm n,\bm b}(x)\prod_{r=1}^k(x-b_r)^{n_s}.
$$
\begin{thm}\label{thm:weyl}
Suppose $\bm b=(b_1,\dots, b_k)$ is a sequence of distinct complex numbers. Then the Gaudin transfer matrix $\mathscr H(x)$ has a simple spectrum in $(\bigotimes_{s=1}^k W_{n_s}(b_s))^\sing$. There exists a bijective correspondence between the monic divisors $y$ of the polynomial $\psi_{\bm n,\bm b}$ and the eigenvectors $v_y$ of the
Gaudin transfer matrix $\mathscr H(x)$ (up to multiplication by a non-zero
constant). Moreover, this bijection is such that \[
\mathscr H(x)v_y = \Big(\frac{1}{2}\zeta_{\bm n,\bm b}'(x)-\zeta_{\bm n,\bm b}(x)\frac{y'(x)}{y(x)}+\frac{1}{2}\big(\zeta_{\bm n,\bm b}(x)\big)^2\Big)v_y.\qedd
\]
\end{thm}
\begin{rem}\label{rem:bae-weyl}
Fix $l\in\Z_{\gge 0}$ and set $\bm t=(t_1,\dots,t_l)$. Let $\bm y_{\bm t}$ represent $\bm t$. Then the Bethe ansatz equation for $(\bigotimes_{s=1}^k W_{n_s}(b_s))^\sing$ is
\[
y_{\bm t}(x) \quad \text{ divides the polynomial } \quad \psi_{\bm n,\bm b}(x).
\]
Note that in this case, $y_{\bm t}$ may have multiple roots. If there are multiple roots in $y_{\bm t}$, then the corresponding on-shell Bethe vector is zero. Therefore an actual eigenvector should be obtained via an appropriate derivative as pointed out in \cite[Section 8.2]{HMVY}.
\qed
\end{rem}

\subsection{The third isomorphism}\label{sec third}
Recall from Section \ref{sec BA}, that without loss of generality, 
we can assume that $\beta_s=0$, $1\lle s\lle k$.  In this case, $\alpha_s=n_s$, $1\lle s\lle k$.

\begin{lem}\label{lem quotient cyc}
There exists a surjective $\g[t]$-module homomorphism from $\bigotimes_{s=1}^k W_{n_s}(b_k)$ to $\bigotimes_{s=1}^k L_{\bla^{(s)}}(b_k)$ which maps vacuum vector to vacuum vector.
\end{lem}
\begin{proof}
It follows from Lemma \ref{lem:weyl-b-more} and our assumption that $\beta_s=0$ for all $1\lle s\lle k$.
\end{proof}
 
By Lemma \ref{lem local weyl}, the surjective $\g[t]$-module homomorphism $$\bigotimes_{s=1}^k W_{n_s}(b_k)\twoheadrightarrow\bigotimes_{s=1}^k L_{\bla^{(s)}}(b_k)$$ induces a surjective $\g[t]$-module homomorphism $$\mc V^\fkS\twoheadrightarrow \bigotimes_{s=1}^k L_{\bla^{(s)}}(b_k).$$ The second map then induces a projection of the Bethe algebras $\mc B_{l}\twoheadrightarrow \mc B_{l,\bLa,\bm b}$. We describe the kernel of this projection. We consider the corresponding ideal in the algebra $\mc O_l$.

Suppose $l\lle k-1$. Define the polynomial $h(x)$ by
\[
h(x)=\prod_{s=1}^k(x-b_s)^{n_s-1}.
\]
Divide the polynomial $g(x)$ in \eqref{eq poly f g} by $h(x)$ and let
\beq\label{eq coeff p}
p(x)=x^{k-l-1}+p_1x^{k-l-2}+\dots+p_{k-l-2}x+p_{k-l-1},
\eeq
\beq\label{eq coeff r}
r(x)=r_{1}x^{n-k-1}+r_2x^{n-k-2}+\dots+r_{n-k-1}x+r_{n-k}
\eeq
be the quotient and the remainder, respectively. Clearly, $p_i,r_j\in\mc O_l$. 

Denote by $I_{l,\bLa,\bm b}^{\mc O}$ the ideal of $\mc O_l$ generated by $r_1,\dots,r_{n-k}$, $\vSi_n- a_n$, and the coefficients of polynomial
\[
\varphi_{\bLa,\bm b}(x)-np(x)f(x)=\prod_{s=1}^k(x-b_s)\sum_{s=1}^k\frac{n_s}{x-b_s}-np(x)f(x).
\]
Let $\mc O_{l,\bLa,\bm b}$ be the quotient algebra
\[
\mc O_{l,\bLa,\bm b}=\mc O_l/I_{l,\bLa,\bm b}^{\mc O}.
\]
Clearly, if $\mc O_{l,\bLa,\bm b}$ is finite-dimensional, then it is a Frobenius algebra. 

Let $I_{l,\bLa,\bm b}^{\mc B}$ be the image of $I_{l,\bLa,\bm b}^{\mc O}$ under the isomorphism $\eta_{l}$.

\begin{lem}\label{lem ann ideal}
The ideal $I_{l,\bLa,\bm b}^{\mc B}$ is contained in the kernel of the projection $\mc B_{l} \twoheadrightarrow \mc B_{l,\bLa,\bm b}$.
\end{lem}
\begin{proof}
We treat $\bm b=(b_1,\dots,b_k)$ as variables. Note that the elements of $I_{l,\bLa,\bm b}^{\mc B}$ act on $\mc M_{l,\bLa,\bm b}$ as polynomials in $\bm b$ with values in $\End((L_{\bLa})_{(n-l,l)}^\sing)$. Therefore it suffices to show it for generic $\bm b$. Let $\mathfrak f(x)$ be the image of $f(x)$ under $\eta_{l}$. The condition that $I_{l,\bLa,\bm b}^{\mc B}$ vanishes is equivalent to the condition that $\varphi_{\bLa,\bm b}(x)$ is divisible by $\mathfrak f(x)$.

By Theorem \ref{thm complete}, there exists an eigenbasis of the operator $\cT(x)$ in $\mc M_{l,\bLa,\bm b}$ for generic $\bm b$. Clearly, a solution of Bethe ansatz equation associated to $\bLa,\bm b, l$ is also a solution to Bethe ansatz equation for $\mc M_{l,\bm a}$, see Theorem \ref{thm:weyl} and Remark \ref{rem:bae-weyl}. Moreover, the expressions of corresponding on-shell Bethe vectors coincide (with different vacuum vectors). By Lemma \ref{lem quotient cyc} and Theorems \ref{thm gl11 eigenvalue in y}, \ref{thm:weyl}, $\varphi_{\bLa,\bm b}(x)$ is divisible by $\mathfrak f(x)$ for generic $\bm b$ since the eigenvalue of $\mathfrak f(x)$ corresponds to $y_{\bm t}(x)$ in \eqref{eq gl11 eigenvalue in y}. Therefore $I_{l,\bLa,\bm b}^{\mc B}$ vanishes for generic $\bm b$, thus completing the proof.
\end{proof}

Therefore, we have the epimorphism
\beq\label{eq epi new}
\mc O_{l,\bLa,\bm b}\cong \mc B_{l} /I_{l,\bLa,\bm b}^{\mc B}\twoheadrightarrow \mc B_{l,\bLa,\bm b}.
\eeq
We claim that the surjection in \eqref{eq epi new} is an isomorphism by checking $\dim \mc O_{l,\bLa,\bm b}=\dim \mc B_{l,\bLa,\bm b}$.

\begin{lem}\label{lem dimen count}
We have $\dim \mc O_{l,\bLa,\bm b}=\displaystyle\binom{k-1}{l}$.
\end{lem}
\begin{proof}
Note that $\C[p_1,\dots,p_{k-l-1},r_1,\dots,r_{n-k}]\cong \C[g_1,\dots,g_{n-l-1}]$, where $p_i$ and $r_j$ are defined in \eqref{eq coeff p} and \eqref{eq coeff r}. It is not hard to check that
\beq\label{eq iso algebra}
\mc O_{l,\bLa,\bm b}\cong \C[f_1,\dots,f_l,p_1,\dots,p_{k-l-1}]/\tilde I_{l,\bLa,\bm b}^{\mc O},
\eeq
where $\tilde I_{l,\bLa,\bm b}^{\mc O}$ is the ideal of $\C[f_1,\dots,f_l,p_1,\dots,p_{k-l-1}]$ generated by the coefficients of the polynomial $\varphi_{\bLa,\bm b}(x)-np(x)f(x)$.

Introduce new variables $\bm w=(w_1,\dots,w_{k-1})$ such that
\[
f(x)=\prod_{i=1}^{l}(x-w_i),\quad p(x)=\prod_{i=1}^{k-l-1}(x-w_{l+i}).
\]Let $\bm \varepsilon=(\varepsilon_1,\dots,\varepsilon_{k-1})$ be complex numbers such that
\[
\varphi_{\bLa,\bm b}(x)=\prod_{s=1}^k(x-b_s)\sum_{s=1}^k\frac{n_s}{x-b_s}=n\Big(x^{k-1}+\sum_{i=1}^{k-1}(-1)^i\varepsilon_ix^{k-1-i}\Big).
\]
Then 
\beq\label{eq iso algebra new}
\C[f_1,\dots,f_l,p_1,\dots,p_{k-l-1}]/\tilde I_{l,\bLa,\bm b}^{\mc O}\cong \C[w_1,\dots,w_{k-1}]^{\fkS_{l}\times \fkS_{k-l-1}}/\langle \sigma_i(\bm w)-\varepsilon_i\rangle_{i=1,\dots,k-1}.
\eeq
The lemma now follows from the fact that $\C[w_1,\dots,w_{k-1}]^{\fkS_{l}\times \fkS_{k-l-1}}$ is a free $\C[w_1,\dots,w_{k-1}]^\fkS$-module of rank $\binom{k-1}{l}$.\end{proof}

Note that we have the projection $(\mc V^\fkS)_{(n-l,l)}^\sing \twoheadrightarrow \mc M_{l,\bLa,\bm b}$. Since by Theorem \ref{thm VS} the Bethe algebra $\mc B_l $ acts on $(\mc V^\fkS)_{(n-l,l)}^\sing$ cyclically, the Bethe algebra $\mc B_{l,\bLa,\bm b}$ acts on $\mc M_{l,\bLa,\bm b}$ cyclically as well. Therefore we have
\beq\label{eq:dim=}
\dim \mc B_{l,\bLa,\bm b}=\dim \mc M_{l,\bLa,\bm b}=\binom{k-1}{l}.
\eeq

\begin{proof}[Proof of Theorem \ref{thm tensor irr}]
Part (i) follows from Lemma \ref{lem dimen count} and \eqref{eq epi new}, \eqref{eq iso algebra}, \eqref{eq iso algebra new}, \eqref{eq:dim=}. Clearly, we have $\mc B_{l,\bLa,\bm b}\cong \mc O_{l,\bLa,\bm b}$ is a Frobenius algebra. Moreover, the map $\rho_l$ from Theorem \ref{thm VS} induces a map
\[
\rho_{l,\bLa,\bm b}:\mc O_{l,\bLa,\bm b}\to \mc M_{l,\bLa,\bm b}
\]
which identifies the regular representation of $\mc O_{l,\bLa,\bm b}$ with the $\mc B_{l,\bLa,\bm b}$-module $\mc M_{l,\bLa,\bm b}$. Therefore part (ii) is proved.

Since $\mc B_{l,\bLa,\bm b}$ is a Frobenius algebra, the regular and coregular representations of the algebra $\mc B_{l,\bLa,\bm b}$ are isomorphic to each other. Parts (iii)--(vi) follow from the general facts about the coregular representations, see e.g. \cite[Section 3.3]{MTV09} or \cite[Lemma 1.3]{Lu20}. 

Due to part (iv), it suffices to consider the algebraic multiplicity of every eigenvalue. It is well known that roots of a polynomial depend continuously on its coefficients. Hence the eigenvalues of $\mc T(x)$ depend continuously on $\bm b$. Part (vii) follows from the deformation argument and Theorem \ref{thm complete}.
\end{proof}

\end{document}